\pdfoutput=1
\documentclass[11pt,a4paper]{article}

\usepackage{amsmath,amssymb,amsthm}
\usepackage{braket}
\usepackage{mathtools}
\usepackage{hyperref}
\usepackage{geometry}
\newtheorem{theorem}{Theorem}[section]
\newtheorem{lemma}[theorem]{Lemma}
\newtheorem{corollary}[theorem]{Corollary}

\newtheorem{definition}[theorem]{Definition}
\newtheorem{remark}[theorem]{Remark}

\DeclareMathOperator{\Tr}{Tr}

\newcommand{\Id}{\mathbb{I}}
\newcommand{\E}{\mathbb{E}}
\newcommand{\C}{\mathbb{C}}

\newcommand{\calF}{\mathcal{F}}

\newcommand{\Uqram}{U_{\mathrm{UQRAM}}}
\newcommand{\Mf}{M_f}
\newcommand{\psif}{\psi_f}
\newcommand{\rhop}{\rho_p}
\newcommand{\Piperp}{\Pi_{\perp}}

\title{Bias-Class Discrimination of Universal QRAM Boolean Memories}
\author{Leonardo Bohac}
\date{}

\begin{document}

\maketitle

\begin{abstract}
We study the discrimination of Boolean memory configurations via a fixed Universal QRAM (U-QRAM) interface. 
Given query access to a quantum memory storing an unknown Boolean function $f:[N]\to\{0,1\}$, 
we ask: what can be inferred about the \emph{bias class} of $f$ (its imbalance from $1/2$, up to complement symmetry) using coherent, addressable queries?
We show that for exact-weight bias classes, the induced \emph{single-query} ensemble state on the address register has a two-eigenspace structure
that yields closed-form expressions for the \emph{single-copy} Helstrom-optimal measurement and success probability.
Because complementing $f$ changes $\ket{\psi_f}$ only by a global phase, hypotheses $p$ and $1-p$ are information-theoretically identical in this model; thus the natural discriminand is the \emph{phase-bias magnitude} $|\mu|$ (equivalently $\mu^2$).
This goes beyond the perfect-discrimination case of Deutsch--Jozsa~\cite{DeutschJozsa92} and complements exact-identification settings such as Bernstein--Vazirani~\cite{BernsteinVazirani97}.
\end{abstract}

\section{Introduction}

Quantum random-access memory (QRAM) has been proposed as a key primitive for quantum algorithms 
that process classical data in superposition~\cite{GLM08}. However, the primary value of QRAM 
is not as a ``fast data loader'' but as a \emph{structured query interface} to quantum memory.

In this work, we study U-QRAM from the perspective of \emph{state discrimination}: 
given query access to an unknown memory configuration via a fixed U-QRAM interface, 
what properties of the memory can be inferred?

This paper continues a program that treats U-QRAM as a \emph{fixed, data-independent interface} whose inputs are different \emph{memory states} (rather than different unitaries), following the construction in~\cite{BohacUqram25}.
It is also complementary to single-query \emph{identification} results for Boolean families (where the goal is to output $f$ itself) developed in~\cite{BohacClosedForm25}.

\paragraph{Why this is not ``just Deutsch--Jozsa.''}
While the optimal single-copy measurement turns out to be as simple as a $\ket{+}$-vs-$\perp$ test, the point of this paper is that this simplicity is \emph{forced by symmetry} of the \emph{U-QRAM induced ensemble} for exact-weight classes.
We isolate this maximally symmetric base case (uniform prior over exact-weight truth tables), obtain the induced mixed states in closed form, and thereby pin down the \emph{information accessible through the fixed interface} as a function of the bias parameter.
This establishes a baseline for follow-ups where the symmetry is broken (approximate weights, noise, non-uniform priors, or genuinely quantum memories), for which Helstrom tests need not remain trivial.

\paragraph{Related work and scope.}
There is prior literature on discriminating \emph{sets} of Boolean functions using quantum-state filtering/discrimination methods (e.g., the Bergou--Hillery line of work), which studies different promise structures and objectives and shares the same core move of mapping a Boolean function to a phase-kickback ``function state'' and then applying quantum-state discrimination/filtering tools~\cite{BergouHillery05}.
Separately, the Deutsch--Jozsa setting has also been framed as discrimination among transformations/oracles in the unitary-discrimination literature~\cite{CollinsDJ}.
Our focus here is the \emph{U-QRAM induced ensemble} for exact-weight bias classes and the resulting \emph{closed-form single-copy} Helstrom test.
For repeated queries we analyze an explicit \emph{separable} strategy under the physically natural ``memory fixed'' model; we do not claim global optimality over all collective or fully coherent oracle algorithms.
In particular, in the standard controlled-oracle query model one can obtain different $\varepsilon$-scalings via amplitude estimation~\cite{BHMT00} (see Remark~\ref{rem:amp-est}).
We view the exact-weight setting as the canonical ``maximal symmetry'' starting point: it is precisely this symmetry that makes the induced states commuting and the Helstrom test explicit.

Our main contributions are:
\begin{enumerate}
    \item[(i)] A reduction from Universal QRAM queries to phase-oracle address states (Lemma~\ref{lem:reduction});
    \item[(ii)] A closed form for the induced ensemble state $\rho_p$ over exact-weight bias classes (Theorem~\ref{thm:ensemble-state});
    \item[(iii)] Optimal measurements and success probabilities for binary bias discrimination in the ensemble model (Theorem~\ref{thm:discrimination});
    \item[(iv)] An achievable multi-query strategy under persistent-memory sampling, with explicit error exponent (Theorem~\ref{thm:multi-query}).
\end{enumerate}

\paragraph{Ensemble hypothesis-testing viewpoint.}
Throughout, we adopt an \emph{ensemble (Bayesian) model} of restricted-access discrimination.
Under a hypothesis $H_p$ (or $H_m$ with $p=m/N$), Nature draws a Boolean function $f$ \emph{once} from a known distribution on the corresponding class (here: uniform on exact-weight truth tables), prepares the preserved memory register in $\ket{M_f}$, and keeps it fixed for the duration of the experiment.
The experimenter cannot access $\ket{M_f}$ directly; they may only query it through the fixed unitary $\Uqram$, producing address states $\ket{\psi_f}$.
As a result, the discrimination task is between the \emph{induced mixed states} $\rho_p$ (averaged over $f$ under the hypothesis), and our single-copy optimality results characterize optimal \emph{average-case} performance rather than worst-case guarantees for every $f$ in a class.
For background on minimum-error quantum hypothesis testing, see~\cite{Helstrom76,Holevo82}.
Moreover, complement symmetry is built in: since $\ket{\psi_{1-f}}=-\ket{\psi_f}$, any procedure in this model can depend on $p$ only through $|\mu|$ (equivalently $\mu^2$).

\section{The Universal QRAM Model}
\label{sec:model}

\subsection{Registers and Memory Encoding}

We consider a U-QRAM system with three quantum registers:
\begin{itemize}
    \item \textbf{Address register} $A$: $n$ qubits, $N = 2^n$ addresses.
    \item \textbf{Data/output register} $D$: 1 qubit (for $K=1$ bit per address).
    \item \textbf{Memory register} $M$: $N$ qubits storing the truth table.
\end{itemize}

A Boolean function $f:[N]\to\{0,1\}$ is encoded in the memory register as the basis state
\begin{equation}
    \ket{\Mf} := \bigotimes_{a=0}^{N-1} \ket{f(a)}_M,
\end{equation}
which we call the \emph{truth-table encoding}.

\subsection{The U-QRAM Unitary}

We model U-QRAM access as a fixed unitary $\Uqram$ acting on $A\otimes D\otimes M$.
This idealized query semantics is consistent with the data-independent U-QRAM interface viewpoint developed in~\cite{BohacUqram25}.

\begin{definition}[U-QRAM query on basis memories]
For any computational-basis memory state
\[
    \ket{m}_M := \bigotimes_{j=0}^{N-1} \ket{m_j}_M,\qquad m_j\in\{0,1\},
\]
the U-QRAM unitary is specified by
\begin{equation}\label{eq:uqram-general}
    \Uqram \ket{a}_A \ket{y}_D \ket{m}_M
    = \ket{a}_A \ket{y \oplus m_a}_D \ket{m}_M ,
\end{equation}
and extended by linearity to arbitrary memory states.
\end{definition}

\begin{remark}[Unitarity]
The map in Eq.~\eqref{eq:uqram-general} permutes computational-basis states and satisfies $\Uqram^2=\Id_{A\otimes D\otimes M}$; hence it extends uniquely to a unitary on $A\otimes D\otimes M$.
\end{remark}

We emphasize that Eq.~\eqref{eq:uqram-general} defines an idealized U-QRAM query for computational-basis memory contents; when $M$ is in superposition over basis memories, the same linear extension generally entangles $A$ and $M$ (cf.\ Remark~\ref{rem:superposed-memories}).

For a Boolean function $f:[N]\to\{0,1\}$ we use the truth-table basis encoding $\ket{\Mf}_M$ defined above,
so that Eq.~\eqref{eq:uqram-general} specializes to
\begin{equation}\label{eq:uqram}
    \Uqram \ket{a}_A \ket{y}_D \ket{\Mf}_M
    = \ket{a}_A \ket{y \oplus f(a)}_D \ket{\Mf}_M .
\end{equation}

\begin{remark}[Key properties]
The unitary $\Uqram$ satisfies:
\begin{enumerate}
    \item[(a)] \textbf{Fixed architecture:} $\Uqram$ is data-independent; different memories correspond to different input states on $M$, not different unitaries.
    \item[(b)] \textbf{Memory preservation:} the memory register is unchanged by a query.
    \item[(c)] \textbf{Address locality:} the data/output depends only on the addressed location.
\end{enumerate}
\end{remark}

\subsection{Reduction to Phase Oracle and Address States}

\begin{lemma}[U-QRAM--Phase-Oracle Reduction]
\label{lem:reduction}
Let $\ket{+} = \frac{1}{\sqrt{N}} \sum_{a=0}^{N-1} \ket{a}$ and $\ket{-} = \frac{1}{\sqrt{2}}(\ket{0} - \ket{1})$.
Then:
\begin{equation}
    \Uqram \big( \ket{+}_A \ket{-}_D \ket{\Mf}_M \big) 
    = \ket{\psif}_A \otimes \ket{-}_D \otimes \ket{\Mf}_M,
\end{equation}
where the \emph{address state} is
\begin{equation}
\label{eq:address-state}
    \ket{\psif} := P_f \ket{+} = \frac{1}{\sqrt{N}} \sum_{a=0}^{N-1} (-1)^{f(a)} \ket{a},
\end{equation}
and $P_f = \sum_a (-1)^{f(a)} \ket{a}\bra{a}$ is the phase oracle.
\end{lemma}

\begin{proof}
Apply $\Uqram$ to the probe state:
\begin{align}
    \Uqram \big( \ket{+}_A \ket{-}_D \ket{\Mf}_M \big)
    &= \frac{1}{\sqrt{N}} \sum_a \Uqram \ket{a} \ket{-} \ket{\Mf} \\
    &= \frac{1}{\sqrt{N}} \sum_a \ket{a} \cdot \frac{1}{\sqrt{2}}\big(\ket{0 \oplus f(a)} - \ket{1 \oplus f(a)}\big) \cdot \ket{\Mf}.
\end{align}
For each $a$:
\[
    \ket{0 \oplus f(a)} - \ket{1 \oplus f(a)} = 
    \begin{cases}
        \ket{0} - \ket{1} = \sqrt{2}\ket{-} & \text{if } f(a) = 0,\\
        \ket{1} - \ket{0} = -\sqrt{2}\ket{-} & \text{if } f(a) = 1.
    \end{cases}
\]
Thus $\ket{-} \mapsto (-1)^{f(a)} \ket{-}$, and
\[
    \Uqram \big( \ket{+} \ket{-} \ket{\Mf} \big) 
    = \bigg( \frac{1}{\sqrt{N}} \sum_a (-1)^{f(a)} \ket{a} \bigg) \otimes \ket{-} \otimes \ket{\Mf}
    = \ket{\psif} \otimes \ket{-} \otimes \ket{\Mf}.
\]
\end{proof}

\begin{corollary}[Conditionally i.i.d.\ repeated queries]
\label{cor:iid}
Conditioned on a fixed (but unknown) $f$ and the truth-table basis encoding $\ket{\Mf}_M$, each query with probe $\ket{+}_A\ket{-}_D$ produces the product output
\[
    \ket{\psif}_A \otimes \ket{-}_D \otimes \ket{\Mf}_M.
\]
Hence, by reinitializing $A,D$ and reapplying $\Uqram$, the experimenter can generate
$\ket{\psif}^{\otimes t}$ on $t$ fresh address registers (conditioned on $f$).
\end{corollary}

\begin{remark}[Gram-matrix viewpoint]\label{rem:gram}
Conditioned on a fixed but unknown $f$, the experimenter receives pure states $\ket{\psi_f}$ whose pairwise overlaps are captured by the Gram matrix $G_{f,g}=\braket{\psi_f}{\psi_g}$.
Gram-matrix methods and related ``square-root''/pretty-good measurements are a standard toolbox for analyzing (near-)symmetric state ensembles~\cite{EldarForney01}.
For a complementary single-query \emph{identification} perspective (recovering $f$ within structured families), see~\cite{BohacClosedForm25}.
In this paper we take the complementary \emph{ensemble-state} route: for a bias class we study the induced mixture $\rho_p=\mathbb{E}_f[\ket{\psi_f}\!\bra{\psi_f}]$, whose symmetry forces a two-eigenspace form and makes the Helstrom test explicit.
\end{remark}

\begin{remark}[Superposed memories]\label{rem:superposed-memories}
If the memory register is prepared in a superposition over truth tables, e.g.\ $\sum_f c_f \ket{\Mf}_M$, then a U-QRAM query generally entangles $A$ and $M$ and the conditional i.i.d.\ picture is no longer automatic.
Analyzing such genuinely quantum-memory effects is a natural direction beyond the present paper.
\end{remark}

\section{Bias Classes and Ensemble States}
\label{sec:bias}

\subsection{Bias as a Function Property}

\begin{definition}[Bias and phase-bias]
For $f:[N]\to\{0,1\}$, define:
\begin{itemize}
    \item The \textbf{bias}: $p(f) := \frac{1}{N} \sum_{a} f(a) = \frac{|f^{-1}(1)|}{N}$.
    \item The \textbf{phase-bias}: $\mu(f) := \frac{1}{N} \sum_a (-1)^{f(a)} = 1 - 2p(f)$.
\end{itemize}
\end{definition}

\begin{remark}[Complement symmetry forces $\mu^2$]
Let $(1-f)(a)=1-f(a)$. Then $\ket{\psi_{1-f}}=-\ket{\psi_f}$, so the induced ensembles for weights $m$ and $N-m$ are identical.
Consequently, any discrimination task between exact-weight classes in this ensemble model can depend on $p$ only through $|\mu|$ (equivalently $\mu^2$).
\end{remark}

\begin{definition}[Exact-weight bias class]
For $m \in \{0, 1, \ldots, N\}$, define
\[
    \calF(m) := \big\{ f:[N]\to\{0,1\} : |f^{-1}(1)| = m \big\}.
\]
This is the set of all Boolean functions with exactly $m$ ones (Hamming weight $m$),
with $|\calF(m)| = \binom{N}{m}$.
\end{definition}

\subsection{Induced ensemble state}

Under the hypothesis $H_m$ we assume $f \sim \mathrm{Unif}(\calF(m))$ is sampled once, the memory is prepared in the corresponding basis state $\ket{M_f}$, and then held fixed.
A single query produces (conditioned on $f$) the address pure state $\ket{\psi_f}$; averaging over the unknown $f$ within the class yields the induced single-copy ensemble state
\begin{equation}
    \rhop := \E_{f \sim \calF(m)} \big[ \ket{\psif}\bra{\psif} \big],
\end{equation}
where $p=m/N$.

\begin{remark}[Ensemble/Bayesian setting]
Throughout, our discrimination results are for the Bayesian/ensemble model induced by the uniform prior over $\calF(m)$.
Worst-case discrimination over all $f\in\calF(m)$ is a distinct problem.
\end{remark}

\begin{theorem}[Ensemble state structure]
\label{thm:ensemble-state}
For the exact-weight bias class $\calF(m)$ with uniform prior, the induced ensemble state is
\begin{equation}
\label{eq:rho-p}
    \rhop = \mu^2 \ket{+}\bra{+} + \frac{1 - \mu^2}{N-1} \Piperp,
\end{equation}
where $\mu = 1 - 2p$, $p = m/N$, and $\Piperp = \Id - \ket{+}\bra{+}$ (with $\Id$ the identity on the address space).
\end{theorem}

\begin{proof}
Let $\mathsf{S}_N$ act on the address basis by permuting labels: for $\pi\in\mathsf{S}_N$, let $U_\pi\ket{a}=\ket{\pi(a)}$.
Because the prior over $\calF(m)$ is uniform and invariant under relabeling of addresses, we have
$\rhop = U_\pi \rhop U_\pi^\dagger$ for all $\pi$.
Therefore $\rhop$ commutes with the full permutation representation on $\C^N$.

The representation decomposes as $\C^N = \mathrm{span}\{\ket{+}\} \oplus \ker(\bra{+})$, and any operator commuting with all $U_\pi$ must act as a scalar on each irreducible component.
This is an instance of Schur's lemma for the permutation representation (see, e.g.,~\cite{FH91}).
Hence there exist scalars $\alpha,\beta$ such that
\begin{equation}
\label{eq:rho-alpha-beta}
    \rhop = \alpha \Id + \beta \ket{+}\bra{+}.
\end{equation}

We determine $\alpha,\beta$ from two constraints.
First, $\Tr(\rhop)=1$ gives $\alpha N + \beta = 1$.
Second,
\[
    \bra{+}\rhop\ket{+} = \E_f\big[|\braket{+}{\psif}|^2\big].
\]
But by \eqref{eq:address-state},
\[
    \braket{+}{\psif} = \frac{1}{N}\sum_{a=0}^{N-1} (-1)^{f(a)} = \mu,
\]
so $\bra{+}\rhop\ket{+} = \mu^2$.
On the other hand, \eqref{eq:rho-alpha-beta} implies $\bra{+}\rhop\ket{+} = \alpha + \beta$.
Thus $\alpha + \beta = \mu^2$.

Solving yields
\[
    \alpha = \frac{1-\mu^2}{N-1},\qquad \beta = \frac{N\mu^2 - 1}{N-1}.
\]
Rewriting in terms of $\Piperp = \Id-\ket{+}\bra{+}$ gives \eqref{eq:rho-p}.
\end{proof}

\begin{corollary}[Two-eigenspace structure]
\label{cor:eigenstructure}
The induced ensemble state $\rho_p$ has:
\begin{itemize}
    \item eigenvalue $\lambda_+ = \mu^2$ with multiplicity 1 (eigenvector $\ket{+}$);
    \item eigenvalue $\lambda_\perp = \frac{1-\mu^2}{N-1}$ with multiplicity $N-1$ (eigenspace $\ker(\bra{+})$).
\end{itemize}
In particular, $\rho_p$ and $\rho_{p'}$ commute for all $p,p'$.
\end{corollary}

\section{Binary Discrimination Between Bias Classes}
\label{sec:discrimination}

Consider two bias classes with parameters $p_0 = m_0/N$ and $p_1 = m_1/N$, with equal prior probability $1/2$ each.

\begin{theorem}[Optimal single-copy discrimination]
\label{thm:discrimination}
For binary discrimination between $\rho_{p_0}$ and $\rho_{p_1}$ with equal priors:
\begin{enumerate}
    \item[(a)] The \textbf{trace distance} is $D(\rho_{p_0}, \rho_{p_1}) = \Delta$, where
    $D(\rho,\sigma)=\tfrac12\|\rho-\sigma\|_1$ and
    \begin{equation}
        \Delta := |\mu_0^2 - \mu_1^2| = |(1-2p_0)^2 - (1-2p_1)^2|.
    \end{equation}

    \item[(b)] An optimal Helstrom measurement is the two-outcome POVM $\{\ket{+}\bra{+}, \Piperp\}$:
    output the hypothesis with larger $\mu^2$ upon outcome $\ket{+}$, and the other upon outcome $\perp$.
    Equivalently: apply $H^{\otimes n}$, measure in the computational basis, and decide based on whether the outcome is $\ket{0^n}$.

    \item[(c)] The \textbf{Helstrom success probability} is
    \begin{equation}
        P_{\mathrm{succ}} = \frac{1 + \Delta}{2}.
    \end{equation}
\end{enumerate}
\end{theorem}

\begin{proof}
By Corollary~\ref{cor:eigenstructure}, $\rho_{p_0}$ and $\rho_{p_1}$ commute and share the eigenbasis $\{\ket{+}\}\cup\ker(\bra{+})$.
For commuting states, the Helstrom optimal measurement is diagonal in the shared eigenbasis, with each projector assigned to the hypothesis with larger eigenvalue.

Moreover,
\[
    \rho_{p_0} - \rho_{p_1} = (\mu_0^2 - \mu_1^2)\ket{+}\bra{+}
    - \frac{\mu_0^2 - \mu_1^2}{N-1}\Piperp,
\]
so the trace norm is
\[
    \|\rho_{p_0} - \rho_{p_1}\|_1
    = |\mu_0^2-\mu_1^2|\Big(1+(N-1)\cdot\frac1{N-1}\Big)
    = 2\Delta.
\]
Thus $D=\Delta$ and $P_{\mathrm{succ}}=\tfrac12(1+D)=\tfrac12(1+\Delta)$.
\end{proof}

\paragraph{Operational recipe (single query).}
\begin{center}
\fbox{\begin{minipage}{0.95\linewidth}
\textbf{Input:} one U-QRAM query access to $\ket{M_f}$ (held fixed), and a fresh address register $A$.
\begin{enumerate}
  \item Prepare $A$ in $\ket{+}$ and $D$ in $\ket{-}$.
  \item Apply $\Uqram$ once to obtain $\ket{\psi_f}$ on $A$ (Lemma~\ref{lem:reduction}).
  \item Apply $H^{\otimes n}$ to $A$ and measure in the computational basis.
  \item If the outcome is $0^n$ (projecting onto $\ket{+}$), decide for the hypothesis with larger $\mu^2$; otherwise decide for the other.
\end{enumerate}
\end{minipage}}
\end{center}

\begin{remark}[Special cases]
\begin{enumerate}
    \item \textbf{Deutsch--Jozsa}~\cite{DeutschJozsa92}: $p_0 = 0$ (constant) vs.\ $p_1 = 1/2$ (balanced).
    Then $\mu_0 = 1$, $\mu_1 = 0$, $\Delta = 1$, $P_{\mathrm{succ}} = 1$ (perfect discrimination).
    \item \textbf{Symmetric biases}: $p_0 = p$ vs.\ $p_1 = 1-p$.
    Then $\mu_0^2=\mu_1^2$, $\Delta=0$, and $P_{\mathrm{succ}}=1/2$ (no discrimination in the ensemble model).
\end{enumerate}
\end{remark}

\section{Multi-Query Performance Under Persistent Memory}
\label{sec:multi-query}

After $t$ queries with the memory held fixed, the experimenter receives (conditioned on $f$) the pure product state $\ket{\psi_f}^{\otimes t}$ on $t$ fresh address registers.
Averaging over $f$ under the hypothesis yields the $t$-copy induced state
\begin{equation}
    \rho_p^{(t)} := \E_{f\sim\calF(m)}\Big[\big(\ket{\psi_f}\bra{\psi_f}\big)^{\otimes t}\Big],
\end{equation}
which need not factor as $\rho_p^{\otimes t}$ in general.

\begin{theorem}[Achievable separable multi-query strategy]
\label{thm:multi-query}
Fix two hypotheses $p_0,p_1$ with $\mu_i=1-2p_i$ and let $t\ge 1$.
Consider the separable measurement that applies the two-outcome POVM $\{\ket{+}\bra{+},\Piperp\}$ independently to each of the $t$ address registers, and then performs a likelihood-ratio test on the resulting count.
Then:
\begin{enumerate}
    \item[(a)] For any $f\in\calF(m)$, one has $\braket{+}{\psi_f}=\mu$, hence the single-shot outcome probability
    \[
      \Pr(+\mid f)=|\braket{+}{\psi_f}|^2=\mu^2
    \]
    is constant across the class.
    \item[(b)] Consequently, under hypothesis $p$ the number $k$ of ``$+$'' outcomes is exactly
    \[
      k\sim\mathrm{Binomial}(t,\mu^2),
    \]
    even though $\rho_p^{(t)}$ may be correlated.
    \item[(c)] For equal priors on $p_0,p_1$, the likelihood-ratio test achieves an error probability bounded by
    \begin{equation}
      P_{\mathrm{err}}(t)\le \tfrac12 e^{-t\,\xi_{\mathrm{cl}}},
    \end{equation}
    where $\xi_{\mathrm{cl}}$ is the classical Chernoff information for distinguishing
    $\mathrm{Bernoulli}(\mu_0^2)$ from $\mathrm{Bernoulli}(\mu_1^2)$:
    \begin{equation}
        \xi_{\mathrm{cl}}
        = -\log \min_{s\in[0,1]}\Big[(\mu_0^2)^s(\mu_1^2)^{1-s} + (1-\mu_0^2)^s(1-\mu_1^2)^{1-s}\Big],
    \end{equation}
    interpreted by continuity when a parameter equals $0$ or $1$.
\end{enumerate}
\end{theorem}

\begin{proof}
Item (a) follows from \eqref{eq:address-state}:
\[
\braket{+}{\psi_f}=\frac1N\sum_{a=0}^{N-1}(-1)^{f(a)}=\mu,
\]
which depends only on the weight $m$ (equivalently $p$), not on the specific $f\in\calF(m)$.

Conditioned on $f$, the $t$ measured outcomes are independent (Corollary~\ref{cor:iid}) and each has success probability $\mu^2$, so $k\mid f\sim\mathrm{Binomial}(t,\mu^2)$.
Since this conditional law does not depend on $f$ within the class, the unconditional law under hypothesis $p$ is also $\mathrm{Binomial}(t,\mu^2)$, proving (b).

Finally, the likelihood-ratio test between two i.i.d.\ Bernoulli models achieves an exponentially decaying error bounded in terms of the Chernoff information~\cite{Chernoff52}, yielding (c).
\end{proof}

\begin{corollary}[Scaling for small bias gap (separable strategy)]
To distinguish $p_0=1/2$ from $p_1=1/2+\varepsilon$ with constant success probability using the separable strategy of Theorem~\ref{thm:multi-query}, it suffices to take
\[
t=\Theta(1/\varepsilon^2).
\]
In the special case $p_0=1/2$ (so $\mu_0^2=0$), the likelihood-ratio test reduces to ``decide $p_1$ iff $k\ge 1$,'' and the error probability is
$P_{\mathrm{err}}(t)=\tfrac12(1-\mu_1^2)^t$.
\end{corollary}

\begin{remark}[Open question: do collective measurements help?]
The $t$-copy induced state $\rho_p^{(t)}$ generally does not factor as $\rho_p^{\otimes t}$, so in principle a collective Helstrom measurement on $\rho_{p_0}^{(t)}$ vs.\ $\rho_{p_1}^{(t)}$ could outperform any separable strategy.
An interesting open problem is to characterize the optimal error exponent for distinguishing these \emph{exchangeable} but potentially correlated mixtures, and to determine whether the ``$\ket{+}$-count'' statistic is sufficient for optimal discrimination in the exact-weight prior model.
\end{remark}

\begin{remark}[Comparison with coherent amplitude estimation]\label{rem:amp-est}
If one allows fully coherent algorithms with repeated controlled uses of the phase oracle $P_f$ (implemented via $\Uqram$ and phase kickback), amplitude estimation can estimate $\mu=\E_a[(-1)^{f(a)}]$ with query complexity scaling as $O(1/\varepsilon)$ in relevant regimes~\cite{BHMT00}.
This comparison assumes the standard query-model capability of \emph{controlled} uses of $P_f$ (or an equivalent controlled-$\Uqram$ construction), as required by amplitude estimation.
Thus the $\Theta(1/\varepsilon^2)$ scaling above should be interpreted as the complexity of the \emph{separable copy-measurement strategy}, not as a general quantum query lower bound.
\end{remark}

\section{Discussion}
\label{sec:discussion}

\subsection{What U-QRAM Contributes}

The discrimination results above could be stated in terms of abstract phase oracles. However, the U-QRAM framing provides:
\begin{enumerate}
    \item \textbf{Operational semantics}: queries are interactions with a fixed, known device, not black-box access to an unknown unitary.
    \item \textbf{Memory persistence}: the hypothesis is ``which memory state?'' and the memory register is preserved by $\Uqram$.
    \item \textbf{Restricted-access viewpoint}: the experimenter never accesses $\ket{M_f}$ directly; the accessible information is mediated by the fixed interface.
\end{enumerate}

\subsection{Relation to Deutsch--Jozsa and Bernstein--Vazirani}

Deutsch--Jozsa~\cite{DeutschJozsa92} is the special case $p_0 = 0$ vs.\ $p_1 = 1/2$ with $\Delta = 1$ (perfect discrimination in the ensemble model).
Bernstein--Vazirani~\cite{BernsteinVazirani97} concerns \emph{identifying} a specific function within a structured family, not discriminating between bias classes; our results are complementary.

\subsection{Limitations and Future Work}

The exact-weight bias class has maximal permutation symmetry, yielding commuting ensemble states and a two-eigenspace form.
Relaxing to approximate-bias classes, noisy QRAM calls, or non-uniform priors breaks this structure and requires more general analyses.
Additionally, for multi-query discrimination under persistent memory, it remains open whether collective measurements across queries can improve the achievable error exponent beyond the explicit separable strategy analyzed here.

\subsection{Scientific Contribution}

This work makes explicit a clean ``restricted-access discrimination'' viewpoint for Universal QRAM: the unknown is not an oracle but a \emph{persisting memory state} accessed only through a fixed, data-independent interface.
Within this model, we introduce bias classes as a first nontrivial hypothesis family and derive the induced single-query ensemble state $\rho_p$ in closed form for exact-weight truth tables, using permutation symmetry to show it has only two eigenspaces.
This yields an immediately implementable Helstrom-optimal single-copy test (and an explicit success probability) that quantitatively characterizes how much class-level information about a Boolean memory is accessible through U-QRAM queries.
More broadly, the paper provides a bridge between Universal QRAM constructions and the quantum state-discrimination toolbox, setting up a systematic path to study less symmetric families, non-uniform priors, noise, and genuinely quantum-memory effects beyond basis-encoded tables.


\end{document}